\documentclass[times, 12pt]{article}

\usepackage[thmnum]{preamble}
\usepackage{enumitem}
\usepackage{bm}
\usepackage{bussproofs}

\newcommand{\CSP}{\ensuremath{\mathrm{CSP}\text{-}\mathrm{SAT}}}
\newcommand{\mCSP}{\ensuremath{\text{CSP}}}
\newcommand{\CC}{\ensuremath{\mathsf{CC}}}

\newcommand{\Search}{\ensuremath{\mathsf{Search}}}
\newcommand{\RCC}{\ensuremath{\mathsf{RCC}}}
\newcommand{\TT}{\ensuremath{\mathrm{TT}}}
\renewcommand{\vars}{\ensuremath{\mathsf{vars}}}

\begin{document}
\title{Random CNFs are Hard for Cutting Planes}
%A Characterization of Semantic Cutting Planes by Monotone Circuits}
\author{Noah Fleming, Denis Pankratov, Toniann Pitassi, Robert Robere}
\maketitle

\begin{abstract}
The random $k$-SAT model is the most important and well-studied 
distribution over $k$-SAT instances.
It is closely connected to
statistical physics; it is used as a testbench for
satisfiablity algorithms, and lastly average-case hardness 
over this distribution has
also been linked to hardness of approximation via Feige's hypothesis.
In this paper, we prove that any Cutting Planes refutation
for random $k$-SAT requires exponential size, for $k$ that is
logarithmic in the number of variables,
and in the interesting regime where the number of clauses 
guarantees that the formula is unsatisfiable with
high probability.
\end{abstract}

\section{Introduction}

The Satisfiability (SAT) problem is perhaps the most famous problem in theoretical computer science, and significant effort has been devoted to understanding randomly generated SAT instances.
The most well-studied random SAT distribution is the random $d$-SAT model, ${\cal F}(m, n, d)$, where a random $d$-CNF over $n$ variables is chosen by uniformly and independently selecting $m$ clauses from the set of all possible clauses on $d$ distinct variables.
The random $d$-SAT model is widely studied for several reasons.
First, it is an intrinsically natural model analogous to the random graph model, and closely related to phase transitions and structural phenomena occurring in statistical physics.
Second, the random $d$-SAT model gives us a testbench of empirically hard examples which are useful for comparing and analyzing SAT algorithms; in fact, some of the better practical ideas in use today originated from insights gained by studying the performance of algorithms on this distribution and the properties of typical random instances.

Third, and most relevant to the current work, the difficulty of solving random $d$-SAT instances above the threshold (in the regime where the formula is
almost certainly unsatisfiable) has recently been connected to worst-case inapproximability \cite{feige-hyp}.
Feige's hypothesis states that there is no efficient algorithm to certify unsatisfiability of random 3-SAT instances for certain parameter regimes of $(m, n, d)$, and he shows that this hard-on-average assumption for 3-SAT implies worst-case inapproximability results for many NP-hard optimization problems.
The hypothesis was generalized to $d$-SAT as well as to any CSP, thus exposing more links to central questions in approximation algorithms and the power of natural SDP algorithms \cite{bks13}.
The importance of understanding the difficulty of solving random $d$-SAT instances in turn makes random $d$-SAT an important family of formulas for propositional proof complexity, since superpolynomial lower bounds for random $d$-SAT formulas in a particular proof system show that {\it any} complete and efficient algorithm based on the proof system will perform badly on random $d$-SAT instances.
Furthermore, since the proof complexity lower bounds hold in the unsatisfiable regime, they are directly connected to Feige's hypothesis.

Remarkably, determining whether or not a random SAT instance from the distribution ${\cal F}(m, n, d)$ is satisfiable is controlled quite precisely by the ratio $\Delta = m/n$, which is called the \emph{clause density}.
A simple counting argument shows that ${\cal F}(m,n,d)$ is unsatisfiable with high probability for $\Delta > 2^d \ln 2$.
The famous satisfiability threshold conjecture asserts that there is a constant $c_d$ such that random $d$-SAT formulas of clause density $\Delta$ are almost certainly satisfiable for $\Delta < c_d$ and almost certainly unsatisfiable if $\Delta > c_d$,
where $c_d$ is roughly $2^d \ln 2$.
In a major recent breakthrough, the conjecture was resolved for large values of $d$ \cite{sat-conj}.

From the perspective of proof complexity, the density parameter $\Delta$ also plays an important role in the \emph{difficulty} of refuting unsatisfiable CNF formulas.
For instance, in Resolution, which is arguably the simplest proof system, the complexity of refuting random $d$-SAT formulas is now very well understood in terms of $\Delta$.
In a seminal paper, Chvatal and Szemeredi \cite{cs-resolution} showed that for any fixed $\Delta$ above the threshold there is a constant $\kappa_{\Delta}$ such that random $d$-SAT requires size $\exp(\kappa_{\Delta} n)$ Resolution refutations with high probability.
%a random $k$-SAT formula with $\Delta \geq ln 2 \cdot 2^k$
%is unsatisfiable and requires exponential size Resolution refutations.
In their proof, the drop-off in $\kappa_{\Delta}$ is doubly exponential in $\Delta$,  making the lower bound trivial when the number of clauses is larger than $n \log^{1/4} n$ (and thus does not hold when $d$ is large.)
Improved lower bounds \cite{BKPS98,bw-simple} proved that the drop-off in $\kappa_{\Delta}$ is at most polynomial in $\Delta$.
More precisely, they prove that a random $d$-SAT formula with at most $n^{(d+2)/4}$ clauses requires exponential size Resolution refutations.
Thus for all values of $d$, even when the number of clauses is way above the threshold, Resolution refutations are exponentially long.
They also give asymptotically matching upper bounds, showing that there are DLL refutations of size $\exp(n/\Delta^{1/(d-2)})$.

Superpolynomial lower bounds for random $d$-SAT formulas are also known for other weak proof systems such as the polynomial calculus and $\mathsf{Res}(k)$ \cite{BI-random,alekhnovich05}, and random $d$-SAT is also conjectured to be hard for stronger semi-algebraic proof systems.
In particular, it is a relatively long-standing open problem to prove superpolynomial size lower bounds for Cutting Planes refutations of random $d$-SAT.
As alluded to earlier, this potential hardness (and even more so for the semi-algebraic SOS proof system) has been linked to hardness of approximation.

In this paper, we focus on the \emph{Chvatal-Gomory Cutting Planes} proof system and some of its generalizations.
%Without loss of generality, we assume that linear integral inequalities are of the form $a^T x \ge b$, where the left-hand side $a^T x$ is the homogeneous part with variables $x = (x_1, \ldots, x_n)$ and coefficients $a \in \mathbb{Z}^n$, and the right-hand side is $b \in \mathbb{Z}$, later referred to as the constant term.
%The most well-studied system for reasoning with integral linear inequalities is \emph{Chvatal-Gomory Cutting Planes (CP)}.
A proof in this system begins with a set of unsatisfiable linear integral inequalities, and new integral inequalities are derived by (i) taking nonnegative linear combinations of previous lines, or (ii) dividing a previous inequality through by 2 (as long as all coefficients on the left-hand side are even) and then rounding up the constant term on the right-hand side.
The goal is to derive the ``false'' inequality $0 \geq 1$ with as few derivation steps as possible.
This system can be generalized in several natural ways.
In \emph{Semantic Cutting Planes}, there are no explicit rules -- a new linear inequality can be derived from two previous ones as long as it follows soundly.
A further generalization of both CP and Semantic CP is the $\CC$-proof system, where now every line is only required to have low (deterministic or real) communication complexity; like Semantic CP, a new line can be derived from two previous ones as long as the derivation is sound.

The main result of this paper is a new proof method for obtaining Cutting Planes lower bounds, and we apply it to prove the first nontrivial lower bounds for the size of Cutting Planes refutations of random $d$-SAT instances.
Specifically we prove that for $d = \Theta(\log n)$ and $m$ in the unsatisfiable regime, with high probability random $d$-SAT requires exponential-size Cutting Planes refutations.
Our main result holds for the other generalizations mentioned above (Semantic CP and $\CC$-proofs).

We obtain the lower bound by establishing an \emph{equivalence} between proving such lower bounds and proving a corresponding monotone circuit lower bound.
Said a different way, we generalize the interpolation method so that it applies to {\it any} unsatisfiable family of formulas.
Namely, we show that proving superpolynomial size lower bounds for any formula for Cutting Planes amounts to proving a monotone circuit lower bound for certain yes/no instances of the monotone CSP problem.
Applying this equivalence to random $d$-SAT instances, we reduce the problem to that of proving a monotone circuit lower bound for a specific family of yes/no instances of the monotone CSP problem.
We then apply the symmetric method of approximations in order to prove exponential monotone circuit lower bounds for our monotone CSP problem.

In recent private communication with Pavel Hrubes and Pavel Pudl{\'{a}}k  we have learned that they have independently proven a similar theorem.

\subsection{Related Work}

Exponential lower bounds on lengths of refutations are known for CP, Semantic CP, and low-weight $\CC$-proofs) \cite{pudlak-cp,FHL16,bpr-cp}
These lower bounds were obtained using the method of interpolation \cite{krajicek-interp}.
A lower bound proof via interpolation begins with a special type of formula -- \emph{an interpolant}.
Given two disjoint $\NP$ sets $U$ and $V$ an interpolant formula has the form $A(x,y) \land B(x,z)$ where the $A$-part asserts that $x \in U$, as verified by the $\NP$-witness $y$, and the $B$-part asserts that $x \in V$, as verified by the $\NP$-witness $z$.
The prominent example in the literature is the clique/coclique formula where $U$ is the set of all graphs with the clique number at least $k$, and $V$ is the set of all $(k-1)$-colorable graphs.
% DENIS: replacing minterm/maxterm definition for a monotone one.
%consisting of a clique of size $k$ and no
%other edges, and $V$ is the set of all graphs consisting of
%a $(k-1)$ coloring; that is, the vertices are partitioned into $k-1$
%groups, with all edges between groups, and no edges within any group.
Feasible interpolation for a proof system amounts to showing that if an interpolant formula has a short proof then we can extract from the proof a small monotone circuit for separating $U$ from $V$.
Thus lower bounds follow from the celebrated monotone circuit lower bounds for clique \cite{razborov-clique,alon-boppana}.

Despite the success of interpolation, it has been quite limited since it only applies to ``split'' formulas.
In particular, the only family of formulas for which are known to be hard for (unrestricted) Cutting Planes are the clique-coclique formulas.
In contrast, for Resolution we have a clean combinatorial characterization for when a formula does or doesn't admit a short Resolution refutation \cite{bw-simple,AD-width}; we would similarly like to understand the strength of Cutting Planes with respect to arbitrary formulas and most notably for random $d$-SAT formulas and Tseitin formulas.

Our main equivalence is an adaptation of the earlier work combined with a key reduction between search problems and monotone functions established in \cite{GP14}.
With this reduction in hand, our main proof is very similar to both \cite{bpr-cp} and \cite{razborov-split}.
\cite{bpr-cp} proved this equivalence for the special case of the
clique-coclique formulas. Namely they showed that low-weight $\CC$-proofs
for this particular formula are equivalent to monotone circuits
for the corresponding sets $U,V$.
Our argument is essentially the same as theirs, only we realize
that it holds much more generally for {\it any} unsatisfiable CNF
and partition of the variables, and the corresponding set of
Yes/No instances of $\mCSP$.

On the other hand, Razborov \cite{razborov-split}
proved the equivalence between PLS communication games (for KW games) and monotone circuits.
The construction in our proof is essentially equivalent to his but
bypasses PLS and proves a direct equivalence between monotone circuits and $\CC$-proofs.
We could have alternatively proven
our equivalence via: (1) Razborov's equivalence between monotone circuits (for a monotone function) and PLS communication games
(for the associated KW game), and then (2) an equivalence between PLS communication games (for a monotone KW game)
and $\CC$-proofs (for the search problem associated with the KW game).
Inspired by \cite{sokolov}, we give a direct argument which is (somewhat) simpler.

%%% Local Variables:
%%% mode: latex
%%% TeX-master: "paper"
%%% End:

\section{Definitions and Preliminaries}

%% Preliminaries

If $x, y \in \set{0,1}^n$ then we write $x \leq y$ if $x_i \leq y_i$ for all $i$.
A function $f: \set{0,1}^n \rightarrow \set{0,1}$ is \emph{monotone} if $f(x) \leq f(y)$ whenever $x \leq y$.
If $f$ is monotone then an input $x \in \set{0,1}^n$ is a \emph{maxterm} of $f$ if $f(x) = 0$ but $f(x') = 1$ for any $x'$ obtained from $x$ by flipping a single bit from $0$ to $1$; dually, $x$ is a \emph{minterm} if $f(x) = 1$ but $f(x') = 0$ for any $x'$ obtained by flipping a single bit of $x$ from $1$ to $0$.
More generally, if $f(x) = 1$ we call $x$ an \emph{accepting instance} or a \emph{yes instance}, while if $f(x) = 0$ then we call $x$ a \emph{rejecting instance} or a \emph{no instance}.
If $x$ is any yes instance of $f$ and $y$ is any no instance of $f$ then there exists an index $i \in [n]$ such that $x_i = 1, y_i = 0$, as otherwise we would have $x \leq y$, contradicting the fact that $f$ is monotone.
If $f, g, h: \set{0,1}^n \rightarrow \set{0,1}$ are boolean functions on the same domain then $f, g \vDash h$ if for all $x \in \set{0,1}^n$ we have $f(x) \wedge g(x) \implies h(x)$.

A \emph{monotone circuit} is a circuit in which the only gates are $\wedge$ or $\vee$ gates.
A \emph{real monotone circuit} is a circuit in which each internal gate has two inputs and computes any function $\phi(x, y) : \reals^2 \rightarrow \reals$ which is monotone nondecreasing in its arguments.

\begin{defn}\label{def:ineq}
A \emph{linear integral inequality} in variables ${x} = (x_1, \ldots, x_n)$ with coefficients ${a} = (a_1, \ldots, a_n) \in \mathbb{Z}^n$ and constant term $b \in \mathbb{Z}$ is an expression 
\[ {a}^T {x} \ge b.\]
\end{defn}

\begin{defn}\label{def:cp_proof}
Given a system of linear integral inequalities $A {x} \ge {b},$ where $A \in \mathbb{Z}^{m \times n}$ and ${b} \in \mathbb{Z}^m$, a \emph{cutting planes proof} of an inequality $a^T {x} \ge c$ is a sequence of inequalities \[ {a_1}^T {x} \ge c_1, {a_2}^T {x} \ge c_2, \ldots, {a_\ell}^T {x} \ge c_\ell,\] such that ${a_\ell} = {a}$, $c_\ell = c$ and every inequality $i \in [\ell]$ satisfies either
% \begin{prooftree}
%   \AxiomC{$a_1x \geq c_1$}
%   \AxiomC{$a_2x \geq c_2$}
%   \BinaryInfC{$(a_1 + a_2)x \geq c_1 + c_2$}
% \end{prooftree}
\begin{itemize}
\item ${a_i}^T {x} \ge c_i$ appears in $A {x} \ge {b}$,
\item ${a_i}^T {x} \ge c_i$ is a Boolean axiom, i.e., $x_j \ge 0$ or $-x_j \ge -1$ for some $j$,
\item there exists $j, k < i$ such that ${a_i}^T {x} \ge c_i$ is the sum of the linear inequalities ${a_{j}}^T {x} \ge c_{j}$ and ${a_{k}}^T {x} \ge c_k$,
\item there exists $j < i$ and a positive integer $d$ dividing every coefficient in $a_j$ such that $a_i = a_j/d$ and $c_i = \lceil c_j/d \rceil$.
\end{itemize}
The \emph{length} of the proof is $\ell$, the number of lines.
If all coefficients and constant terms appearing in the cutting planes proof are bounded by $O(\poly(n))$, then the proof is said to be of \emph{low weight}.
\end{defn}

Let ${\cal F} = C_1 \land \ldots \wedge C_m$ be an unsatisfiable CNF formula over variables $z_1, \ldots, z_n$.
For any clause $C$ let $C^-$ denote the set of variables appearing negated in the clause and let $C^+$ denote variables occurring positively in the clause.
Each clause $C$ in ${\cal F}$ can be encoded as a linear integral inequality as \[ \sum_{z_i \in C^+} z_i + \sum_{z_i \in C^-} (1-z_i) \ge 1.\]
Thus each unsatisfiable CNF can be translated into a system of linear integral inqualities $Az \geq b$ with no $0/1$ solutions.
A \emph{cutting planes (CP) refutation} of this system is a cutting planes proof of the inequality $0 \ge 1$ from $A x \ge b$.

\begin{defn}\label{def:semantic-proof}
  Let ${\cal F} = C_1 \land \ldots \wedge C_m$ be an unsatisfiable $k$-CNF on $n$ variables.
  A \emph{semantic refutation} of ${\cal F}$ is a sequence \[ L_1, L_2, \ldots, L_\ell\] of boolean functions $L_i: \set{0,1}^n \rightarrow \set{0,1}$ such that
  \begin{enumerate}
  \item $L_i = C_i$ for all $i = 1, 2, \ldots, m$.
  \item $L_\ell = 0$, the constant $0$ function.
  \item For all $i > m$ there exists $j, k < i$ such that $L_j, L_k \vDash L_i$.
  \end{enumerate}
  The \emph{length} of the refutation is $\ell$.
\end{defn}

We will be particularly interested in semantic refutations where the boolean functions can be computed by short communication protocols.
\begin{defn}\label{def:semantic-cc}
  Let ${\cal F} = C_1 \land \ldots \wedge C_m$ be an unsatisfiable CNF on $n = n_1 + n_2$ variables, and let $X = \set{x_1, x_2, \ldots, x_{n_1}}$, $Y= \set{y_1, \ldots, y_{n_2}}$ be a partition of the variables.
  A $\CC_k$-refutation of ${\cal F}$ with respect to the partition $(X, Y)$ is a semantic refutation \[ L_1, \ldots, L_\ell \] of ${\cal F}$ such that each function $L_i$ in the proof can be computed by a $k$-bit communication protocol with respect to the partition $(X, Y)$.
\end{defn}

%It can be shown that semantic-CC refutations of cost $d$ and size $m$
%are equivalent to PLS communication games (of cost roughtly $d$ and size $m$).
%However, we will work directly with semantic-CC refutations.

Since any linear integral inequality $ax + by \geq c$ with polynomially bounded weights can be evaluated by a trivial $O(\log n)$-bit communication protocol (just by having Alice evaluating $ax$ and sending the result to Bob), it follows that low-weight cutting planes proofs are also $\CC_{O(\log n)}$-proofs.
We can similarly define a proof system which can simulate any cutting planes proof by strengthening the type of communication protocol.

\begin{defn}\label{def:real-communication}
  A $k$-round \emph{real communication protocol} is communication protocol between two players, Alice and Bob, where Alice receives an input $x \in {\cal X}$ and Bob receives $y \in {\cal Y}$.
  In each round, Alice and Bob each send real numbers $\alpha, \beta$ to a ``referee'', who responds with a single bit $b$ which is $1$ if $\alpha \geq \beta$ and $0$ otherwise.
  After $k$ rounds of communication, the players output a bit $b$.
  The protocol computes a function $F: {\cal X} \times {\cal Y} \rightarrow \set{0,1}$ if for all $(x, y) \in {\cal X} \times {\cal Y}$ the protocol outputs $F(x, y)$.
\end{defn}

\begin{defn}\label{def:real-cc}
  Let ${\cal F} = C_1 \land \ldots \wedge C_m$ be an unsatisfiable CNF on $n = n_1 + n_2$ variables $X = \set{x_1,\ldots,x_{n_1}}$ and $Y = \set{y_1,\ldots y_{n_2}}$.
  An $\RCC_k$-refutation of ${\cal F}$ is a semantic refutation \[ L_1, L_2, \ldots, L_\ell \] in which each function $L_i$ can be computed by a $k$-round real communication protocol with respect to the variable partition $X, Y$.
\end{defn}

It is clear that \emph{any} linear integral inequality $ax + by \geq c$ can be evaluated by a $1$-round real communication protocol, and so it follows that a cutting planes refutation of ${\cal F}$ is also an $\RCC_1$-refutation of ${\cal F}$.
We record each of these observations in the next proposition.

\begin{prop}\label{prop:cp-cc}
  Let ${\cal F}$ be an unsatisfiable CNF on variables $z_1, z_2, \ldots, z_n$, and let $X, Y$ be any partition of the variables into two sets.
  Any length-$\ell$ low-weight cutting planes refutation of ${\cal F}$ is a length-$\ell$ $\CC_{O(\log n)}$-refutation of ${\cal F}$.
  Similarly, any length-$\ell$ cutting planes refutation of ${\cal F}$ is a length-$\ell$ $\RCC_1$-refutation of ${\cal F}$.
\end{prop}

\subsection{Total Search Problems and Monotone CSP-SAT}\label{sec:csp-sat}

In this section we review the equivalence between the search problem associated with an unsatisfiable CNF formula, and the Karchmer-Wigderson (KW) search problem for a related (partial) monotone function.

\begin{defn}\label{def:total-search}
  Let $n_1, n_2, m$ be positive integers, and let ${\cal X}, {\cal Y}$ be finite sets.
  A \emph{total search problem} is a relation ${\cal R} \subseteq {\cal X}^{n_1} \times {\cal Y}^{n_2} \times [m]$ where for each $(x,y) \in {\cal X}^{n_1} \times {\cal Y}^{n_2}$, there is an $i \in [m]$ such that ${\cal R}(x,y,i)=1$.
  We refer to $x \in {\cal X}^{n_1}$ as Alice's input and $y \in {\cal Y}^{n_2}$ as Bob's input.
  The search problem is \emph{$d$-local} if for each $i \in [m]$ we have that ${\cal R}(*,*,i)$ depends on a fixed set of at most $d$ coordinates of $x$ (it may depend on any number of $y$ coordinates).
\end{defn}

A standard example of a $d$-local search problem is the search problem associated with unsatisfiable $d$-CNFs.
\begin{defn}\label{def:canonical-search}
  Let ${\cal F}$ be an unsatisfiable $d$-CNF formula with $m$ clauses and $n$ variables $z_1,\ldots,z_{n}$.
  Consider any partition of $z_1, z_2, \ldots, z_n$ into two sets $x_1, x_2, \ldots, x_{n_1}$ and $y_1, y_2, \ldots, y_{n_2}$.
  The search problem $\Search({\cal F})$ with respect to this partition takes as input an assignment $x \in \set{0,1}^{n_1}$ and $y \in \set{0,1}^{n_2}$ and outputs the index $i \in [m]$ of a violated clause under this assignment.
\end{defn}

This problem is clearly $d$-local since each clause can contain at most $d$ variables from $x_1, x_2, \ldots, x_{n_1}$.
Associated with this search problem is the following monotone variant of the constraint satisfaction problem.

\begin{defn}
  Let $H  = (L \cup R, E)$ be a bipartite graph such that each vertex $v \in L$ has degree at most $d$, and let $m = |L|$ and $n = |R|$.
  Let $\Sigma$ be a finite alphabet.
  A \emph{constraint satisfaction problem} (CSP) ${\cal H}$ with topology $H$ and alphabet $\Sigma$ is defined as follows.
  The vertices in $L$ are thought of as the set of \emph{constraints}, and the vertices in $R$ are thought of as a set of \emph{variables}; thus for each vertex $i \in L$ we let $\vars(i)$ denote the neighbourhood of $i$.
  For each vertex $i \in L$ the CSP has an associated boolean function $\TT_u: \Sigma^{\vars(i)} \rightarrow \set{0,1}$ called the \emph{truth table} of $i$ that encodes the set of ``satisfying'' assignments to the constraint associated with $i$.
  An assignment $\alpha \in \Sigma^n$, thought of as a $\Sigma$-valued assignment to the variables $R$, \emph{satisfies} the CSP ${\cal H}$ if for each $i \in L$ we have $\TT_i(\alpha \restriction \vars(i)) = 1$, otherwise the assignment \emph{falsifies} the CSP.

  For each $i \in [m]$ and $\alpha \in \Sigma^{\vars(i)}$ we abuse notation and let $\TT_i(\alpha)$ represent the boolean variable corresponding to this entry of the truth table for the constraint $i$.
\end{defn}

\begin{defn}\label{def:csp-sat}
  Let $H  = (L \cup R, E)$ be a bipartite graph such that each vertex $i \in L$ has degree at most $d$, and let $m = |L|$ and $n = |R|$.
  We think of $H$ as encoding the topology of a constraint satisfaction problem, where each vertex $i \in L$ represents a \emph{constraint} of the CSP and each $i \in R$ represents a \emph{variable} of the CSP.
  Let $\Sigma$ be a finite alphabet, and let $N = \sum_{i = 1}^m |\Sigma|^{\vars(i)} \leq m|\Sigma|^d$.
  The monotone function $\CSP_{H, \Sigma} : \set{0,1}^N \rightarrow \set{0,1}$ is defined as follows.
  An input $x \in \set{0,1}^N$ encodes a CSP ${\cal H}(x)$ by specifying for each vertex $u \in L$ its truth table \[\TT_u^x : \Sigma^{\vars(u)} \rightarrow \set{0,1}.\]
  Given an assignment $x \in \set{0,1}^N$ the function $\CSP_{H, \Sigma}(x) = 1$ if and only if the CSP ${\cal H}(x)$ is satisfiable.
  This function is clearly monotone since for any $x, y \in \set{0,1}^N$ with $x \leq y$, any satisfying assignment for the CSP ${\cal H}(x)$ is also a satisfying assignment for the CSP ${\cal H}(y)$.
\end{defn}

Next we show how to relate $d$-local total search problems and the $\CSP$ problem.
Let ${\cal R} \subseteq {\cal X}^{n_1} \times {\cal Y}^{n_2} \times [m]$ be a $d$-local total search problem.
Associated with ${\cal R}$ is a bipartite \emph{constraint graph} $H_{\cal R}$ encoding for each $i \in [m]$ the coordinates in ${\cal X}^{n_1}$ on which ${\cal R}(*, *, i)$ depends.
Formally, the constraint graph is the bipartite graph $H_{\cal R} = (L \cup R, E)$ with $L = [m]$, $|R| = [n_1]$, and for each pair $(i, j) \in L \times R$ we add the edge if ${\cal R}(*, *, i)$ depends on the variable $x_j$.
Note that each vertex $u \in L$ has degree at most $d$, since the original search problem is $d$-local.

Given ${\cal R}$ and its corresponding constraint graph we can give a natural way to construct accepting and rejecting instances of $\CSP_{H_{\cal R}, {\cal X}}$ from ${\cal X}^{n_1}$ and ${\cal Y}^{n_2}$.
To reduce clutter, given a $d$-local total search problem ${\cal R}$ we abuse notation and write \[\CSP_{{\cal R}} := \CSP_{H_{\cal R}, {\cal X}}.\]

\begin{description}
\item[Accepting Instances ${\cal U}$.] For any $x \in {\cal X}^{n_1}$ we construct an accepting input ${\cal U}(x)$ of $\CSP_{{\cal R}}$ as follows.
  For each vertex $i \in L$ we define the corresponding truth table $\TT_i$ by setting $\TT_i(\alpha) = 1$ if $x \restriction \vars(i) = \alpha$ and $\TT_i(\alpha) = 0$ otherwise.
\item[Rejecting Instances ${\cal V}$.] For any $y \in {\cal Y}^{n_2}$ we construct a rejecting input ${\cal V}(y)$ of $\CSP_{{\cal R}}$ as follows.
  For each vertex $i \in L$ and each $\alpha \in \Sigma^{\vars(i)}$ we set \[\TT_i(\alpha) = 0 \Longleftrightarrow {\cal R}(\alpha, y, i) \text{ holds}.\]
\end{description}

Given $x \in {\cal X}^{n_1}$ it is easy to see that ${\cal U}(x)$ is a satisfying assignment for $\CSP_{\cal R}$ since $x$ is a satisfying assignment for the corresponding CSP.
The rejecting instances require a bit more thought.
Let $y \in {\cal Y}^{n_2}$ and consider the rejecting instance ${\cal V}(y)$ as defined above.
Suppose by way of contradiction that the corresponding CSP ${\cal H}_{\cal R}({\cal V}(y))$ is satisfiable, and let $x \in {\cal X}^{n_1}$ be the satisfying assignment for the CSP.
It follows by definition of the rejecting instances that ${\cal R}(x, y, u)$ does not hold for any $u$, implying that ${\cal R}$ is not total.

\section{Relating Proofs and Circuits}
In this section we relate $\CC_d$-proofs and monotone circuits, as well as $\RCC_1$-proofs and \emph{real} monotone circuits.

\begin{thm}\label{thm:main_theorem}
  Let ${\cal F}$ be an unsatisfiable CNF formula on $n$ variables and let $X = \set{x_1, \ldots, x_{n_1}}$, $Y = \set{y_1, \ldots, y_{n_2}}$ be any partition of the variables.
  Let $k$ be a positive integer.
  If there is a $\CC_k$ refutation of ${\cal F}$ with respect to the partition $(X, Y)$ of length $\ell$, then there is a monotone circuit separating the accepting and rejecting instances ${\cal U}(\set{0,1}^{n_1}), {\cal V}(\set{0,1}^{n_2})$ of $\CSP_{\Search({\cal F})}$ of size $O(2^k \ell)$.
\end{thm}
\begin{proof}
  Let ${\cal F} = C_1 \wedge \ldots \wedge C_m$ over variables $x_1,\ldots,x_{n_1},y_1,\ldots,y_{n_2}$.
  Let $P$ be a $\CC_k$-proof for ${\cal F}$ with $\ell$ lines.
  Order the lines in $P$ as $L_1,L_2, \ldots, L_{\ell}$, where the final line $L_{\ell}$ is the identically false formula, and each earlier line is either a clause, or follows semantically from two earlier lines. 

  We build the circuit for $\CSP_{\Search({\cal F})}$ that separates ${\cal U}, {\cal V}$ by induction on $\ell$.
  For each line $L$ in the proof, there are $2^k$ possible histories $h$, each with an associated monochromatic rectangle $R_L(h)$.
  A rectangle $h$ is {\it good} for $L$ if it is $0$-monochromatic.
  For every line $L$ and each good history $h$ for $L$, we will build a circuit ${\cal C}^L_h$ that correctly ``separates'' $x$ and $y$ for each $(x,y) \in R_L(h)$.
  By this, we mean that the circuit ${\cal C}^L_h$ outputs $1$ on ${\cal U}(x)$ (the 1-input associated with $x$) and outputs $0$ on ${\cal V}(y)$ (the $0$-input associated with $y$).

  For each leaf in the proof, the associated line $L$ is a clause $C_i$ of ${\cal F}$. 
  The communication protocol for $C_i$ is a two-bit protocol where Alice/Bob each send $0$ iff their inputs are $\alpha,\beta$ such that $C_i(\alpha,\beta)=0$.
  Thus there is only one good (0-monochromatic) rectangle with history $h=00$.
  This pair $\alpha,\beta$ corresponds to the variable $\TT_i(\alpha)$, and we define the circuit ${\cal C}^L_h$ corresponding to line $L=C_i$ and good history $h=00$ to be the variable $\TT_i(\alpha)$.
  % (TONI NOTE: If it is a lifted formula and therefore the $C_i$'s aren't clauses
  % but constraints, then we need to fix this up a bit.)

  Now suppose that $L$ is derived from $L_1$ and $L_2$, and inductively we have circuits ${\cal C}^{L_1}_{h'}$, ${\cal C}^{L_2}_{h''}$ for each history $h'$ good for $L_1$ and $h''$ good for $L_2$.
  Given a good history $h$ for $L$, we will show how to build the circuit ${\cal C}^L_h$.
  It will use all of the the circuits that were built for $L_1$ and $L_2$ ($\{{\cal C}^{L_1}_{h'}, {\cal C}^{L_2}_{h''}\}$ for all good $h'$) and an additional $2^k$ gates.
  To build ${\cal C}^L_h$ we will construct a {\it stacked} protocol tree for $L$, corresponding to first
  running the communication protocol for $L_1$ and then running the communication protocol for $L_2$.
  This will give us a height $2k$ (full) binary tree, $T$, where the top part is the communication protocol tree for $L_1$, with protocol trees for $L_2$ hanging off of each of the leaves.
  We label each of the leaves of this stacked tree with a circuit from $\{{\cal C}^{L_1}_{h'}, {\cal C}^{L_2}_{h''}\}$ as follows.
  Consider a path labelled $h_1 h_2$ in $T$, where $h_1$ is the history from running $L_1$ and $h_2$ is the history from running $L_2$.
  By soundness, either the rectangle $R_L(h) \cap R_{L_1}(h_1)$ is $0$-monochromatic, or the rectangle $R_L(h) \cap R_{L_2}(h_2)$ is $0$-monochromatic.
  In the first case, we will label this leaf with ${\cal C}^{L_1}_{h_1}$ and otherwise we will label this leaf with ${\cal C}^{L_2}_{h_2}$.
  Now we will label the internal vertices of the stacked tree with a gate: if a node corresponds to Alice speaking, then we label the node with an $\lor$ gate, and otherwise if the node corresponds to Bob speaking, then we label the node with an $\land$ gate.
  The resulting circuit has size $2^k$ plus the sizes of the subcircuits, and thus the total circuit size is $2^k \ell$.
  The theorem is therefore immediately implied by the following claim.
  
  \spcnoindent
  {\bf Claim. }The circuit resulting from the above construction satisfies: for each line $L$ in $P$, and for each good history $h$ for $L$, ${\cal C}^L_h$ will be correct for all $(x,y) \in R_L(h)$.

  \spcnoindent
  \emph{Proof of Claim.}
    If $L$ is an axiom, then $L$ is a clause, $C_i$.
    The communication protocol for $C_i$ is a two-bit protocol where Alice and Bob each send $0$ iff their part of $C_i$ evalutes to 0.
    There is only one good (0-monochromatic) history, $h=00$.
    If $(x,y) \in R_L(h)$ then $C_i(x, y)=0$ by definition.
    Let $\alpha = x \restriction \vars(C_i)$.
    In our construction the circuit corresponding to ${\cal C}^L_{h}$ is labelled by the variable $\TT_i(\alpha)$, and it is easy to check that ${\tilde x}$ sets $\TT_i(\alpha)$ to true, and ${\tilde y}$ sets $\TT_i(\alpha)$ to false.

    If $L$ is not an axiom, then we will prove the
    lemma by proving the following stronger statement
    by induction: For each line $L$ (derived from previous lines $L_1$ and $L_2$), 
    and for each node $v$ in the stacked protocol tree for $L$,
    with corresponding (sub)history 
    $h' = h_1 h_2$, the subcircuit ${\cal C}^L_{h'}$ associated with vertex $v$
    is correct on all $(x,y) \in R_L(h) \cap R_{L_1}(h_1) \cap R_{L_2}(h_2)$.
    % (By correct on a pair $(x,y)$ we mean that $C^{h'}_L({\overline x}) =1$ and $C^{h'}_L({\overline y})=0$.
    % Also we will associate histories with subrectangles, so by $h \cap h'$ we mean the subrectangle
    % that is the intersection of the rectangle defined by $h$ and the rectangle defined by $h'$.)

    Fix a line $L$ that is not an axiom.
    For the base case, suppose that $v$ is a leaf of the stacked protocol tree for $L$ with 
    history $h' = h_1 h_2$.
    Then by soundness either (i) $R_L(h) \cap R_{L_1}(h_1) =0$ or (ii) $R_L(h) \cap R_{L_1}(h_2) =0$.
    In case (i) we labelled $v$ by ${\cal C}^{L_1}_{h_1}$. 
    Since $R_L(h) \cap R_{L_1}(h_1) =0$, $R_{L_1}(h_1) =0$ and
    therefore ${\cal C}^{L_1}_{h_1}$ is defined and is correct on all $(x,y) \in R_{L_1}(h_1)$, so it is
    correct on all $(x,y) \in R_L(h) \cap R_{L_1}(h_1) \cap R_{L_2}(h_2)$. A similar argument holds in case (ii).

    For the inductive step, let $v$ be a nonleaf node in the protocol tree with history $h'$
    and assume that Alice owns $v$. 
    The rectangle $R_L(h) \cap R_{L_1}(h_1) \cap R_{L_2}(h_2) = A \times B$ is partitioned into $A_0 \times B$ and $A_1 \times B$, where
    \begin{enumerate}
    \item $A = A_0 \cup A_1$,
    \item $A_0 \times B$ is the rectangle with history $h'0$,
    \item $A_1 \times B$ is the rectangle with history $h' 1$.
    \end{enumerate}
    Given $(x,y) \in R_L(h) \cap R_{L_1}(h_1) \cap R_{L_2}(h_2)$, since  ${\cal C}^{L}_{h'0}$ is correct on all $(x,y) \in A_0 \times B$ and ${\cal C}^{L}_{h'1}$ is correct on all $(x,y) \in A_1 \times B$, it follows that ${\cal C}^L_h = {\cal C}^{L}_{h'0} \lor {\cal C}^{L}_{h'1}$ is correct on all $(x,y) \in A \times B$. 
    To see this, observe that if $x \in A_0$, then ${\cal C}^{L}_{h'0}({\cal U}(x))=1$ and therefore \[{\cal C}^{L}_h({\cal U}(x)) = {\cal C}^{L}_{h'0}({\cal U}(x)) \lor {\cal C}^{L}_{h'1}({\cal U}(x))=1.\]
    Similarly, if $x \in A_1$, then ${\cal C}^{L}_{h'1}({\cal U}(x))=1$ and therefore \[{\cal C}^L_h({\cal U}(x)) = {\cal C}^{L}_{h'0}({\cal U}(x)) \lor {\cal C}^{L}_{h'1}({\cal U}(x))=1.\]
    Finally if $y \in B$ then both ${\cal C}^{L}_{h'0}({\cal V}(y))={\cal C}^{L}_{h'1}({\cal V}(y))=0$ and therefore \[{\cal C}^{L}_h({\cal V}(y)) = {\cal C}^{L}_{h'0}({\cal V}(y)) \lor {\cal C}^{L}_{h'1}({\cal V}(y))=0.\]
    A similar argument holds if $v$ is an internal node
    in the protocol tree that Bob owns (and is therefore labelled by an AND gate.
\end{proof}
  
The converse direction is much easier.

\begin{thm}
  If there is a monotone circuit separating these inputs of $\CSP_{\Search({\cal F})}$ of size $\ell$, then there is a $\CC_2$-refutation of ${\cal F}$ of length $\ell$ with respect to this partition of the variables.
\end{thm}
\begin{proof}
  In the other direction, we show that from a small monotone circuit ${\cal C}$ for $\CSP_{\Search({\cal C})}$ that separates ${\cal U}(\set{0,1}^{n_1})$ and ${\cal V}(\set{0,1}^{n_2})$, we can construct a small $\CC_2$-proof for ${\cal F}$, where Alice gets $x \in \set{0,1}^{n_1}$ and Bob gets $y \in \set{0,1}^{n_2}$.
  The lines/vertices of the refutation will be in 1-1 correspondence with the gates of ${\cal C}$. 
  % Each OR gate will correspond to
  % a protocol where Alice sends just a single bit, and each AND gate will correspond
  % to Bob sending 1 bit. 
  The protocol is constructed inductively from the leaves of ${\cal C}$ to the root.
  For a gate $g$ of ${\cal C}$, let $U_g$ be those inputs $u \in {\cal U}(\set{0,1}^{n_1})$ such that $g(u)=1$, and let $V_g$ be those inputs $v \in {\cal V}(\set{0,1}^{n_2})$ such that $g(v)=0$.
  At each gate $g$ we will prove that for every pair $(u, v) \in U_g \times V_g$ and for every $(x, y)$ such that $u = {\cal U}(x), v = {\cal V}(y)$, the protocol $R_g$ on input $(x, y)$ will output 0.
  % solves the search problem for the pair $(x,y)$.
  Since the output gate of ${\cal C}$ is correct for all pairs, this will achieve our desired protocol.

  At a leaf $\ell$ labelled by some variable $\TT_{j}(\alpha)$, the pairs associated with this leaf must have $\TT_j(\alpha) = 1$ in $u$ and $0$ in $v$, and thus we can define $R_{\ell}(x, y)$ to be 0 if and only if $x$ is consistent with $\alpha$ and the clause $C_j$ evaluates to false on $(x,y)$.
  This is a 2-bit protocol, and by definition of the accepting and rejecting instances we have for all $(x, y)$ satisfying $u = {\cal U}(x), v = {\cal V}(y)$ that $x \restriction \vars(j) = \alpha$ and ${\cal R}(\alpha, y, j)$ holds.

  Now suppose that $g$ is a OR gate of ${\cal C}$, with inputs $g_1, g_2$.
  The protocol $R_g$ on $(x, y)$ is as follows. 
  Alice privately simulates ${\cal C}_{g_1}({\cal U}(x))$ and ${\cal C}_{g_2}({\cal U}(x))$, and Bob simulates ${\cal C}_{g_1}({\cal V}(y))$ and ${\cal C}_{g_2}({\cal V}(y))$.
  If (i) either ${\cal C}_{g_1}({\cal U}(x)) =1$ or ${\cal C}_{g_2}({\cal U}(x))=1$ and (ii) both ${\cal C}_{g_1}({\cal V}(y))=0$ and ${\cal C}_{g_2}({\cal V}(y))=0$, then they output 0, and otherwise they output 1.
  This is a 2-bit protocol, with Alice sending one bit to report whether or not condition (i) is satisfied, and Bob sending one bit to report if (ii) is satisfied.

  Now, we want to show that for all $(x,y)$ such that ${\cal C}_g({\cal U}(x)) = 1$ and ${\cal C}_g({\cal V}(y)) = 0$ we have that $R_{g}(x,y)=0$.
  This is easy --- since $g = g_1 \vee g_2$ we have that ${\cal C}_g({\cal U}(x)) = 1$ and ${\cal C}_g({\cal V}(y) = 0$ implies that either ${\cal C}_{g_1}({\cal U}(x)) =1$ or ${\cal C}_{g_2}({\cal U}(x))=1$ and ${\cal C}_{g_1}({\cal V}(y))=0$ and ${\cal C}_{g_2}({\cal V}(y))=0$, implying that the protocol will output $0$ on $(x, y)$ by definition.
  % Then, assuming that $R_g(x,y)=0$, it follows that either (a) ${\cal C}_{g_1}({\cal U}(x))=1$ and ${\cal C}_{g_1}({\cal V}(y))=0$, or (b) ${\cal C}_{g_2}({\cal U}(x))=1$ and ${\cal C}_{g_2}({\cal V}(y))=0$, by conditions (i) and (ii).

  % In case (a), by the induction hypothesis, $R_{g_1}(x,y)=0$
  % and in case (b) by the induction hypothesis, $R_{g_2}(x,y)=0$.
  % It is also clear that since $g$ is an OR gate, $U_g = U_{g_1} \cup U_{g_2}$ and $V_g = V_{g_1} \cap V_{g_2}$,
  % and that $R_g$ outputs 0 on all pairs in $U_g \times V_g$.

  Similarly, if $g$ is an AND gate, then again Alice privately simulates ${\cal C}_{g_1}({\cal U}(x))$ and ${\cal C}_{g_2}({\cal U}(x))$ and Bob privately simulates ${\cal C}_{g_2}({\cal V}(y))$ and ${\cal C}_{g_2}({\cal V}(y))$.
  If (i) ${\cal C}_{g_1}({\cal U}(x))=1$ and ${\cal C}_{g_2}({\cal U}(x))=1$ and (ii) either ${\cal C}_{g_2}({\cal V}(y))=0$ or ${\cal C}_{g_2}({\cal V}(y))=0$, then they ouput 0, and otherwise they output 1.
  By an analogous argument to the OR case, it's easy to see that the protocol will output $0$ whenever ${\cal C}_g({\cal U}(x)) = 1$ and ${\cal C}_g({\cal V}(y)) = 0$.
\end{proof}

The following theorem was recently proven \cite{pp-unpublished}, showing that
$\RCC_1$-proofs imply monotone real circuits for the associated
search problem.

\begin{thm}\cite{pp-unpublished}\label{thm:main_theorem}
  Let ${\cal F}$ be an unsatisfiable CNF formula on $n$ variables and let $X = \set{x_1, \ldots, x_{n_1}}$, $Y = \set{y_1, \ldots, y_{n_2}}$ be any partition of the variables.
  If there is a $\RCC_1$ refutation of ${\cal F}$ with respect to the partition $(X, Y)$ of length $\ell$, then there is a monotone real circuit separating the accepting and rejecting instances ${\cal U}(\set{0,1}^{n_1}), {\cal V}(\set{0,1}^{n_2})$ of $\CSP_{\Search({\cal F})}$ of size polynomial in $\ell$.
\end{thm}

In particular, the above theorem implies that for any family of formulas ${\cal F}$
and for any partition of the underlying variables into $X,Y$, 
a Cutting Planes refutation of ${\cal F}$ of size $S$ implies a 
similar size monotone real circuit for separating the accepting and rejecting
instances  ${\cal U}(\set{0,1}^{n_1}), {\cal V}(\set{0,1}^{n_2})$ of $\CSP_{\Search({\cal F})}$. 

%% Local Variables:
%%% mode: latex
%%% TeX-master: "paper"
%%% End:

\section{Lower Bounds for Random CNFs}
%% Random CNFs

In this section using Theorem \ref{thm:main_theorem} we prove lower bounds for $\RCC_1$-refutations (and therefore cutting planes refutations) of uniformly random $d$-CNFs with sufficient clause density.

\begin{defn}
  Let ${\cal F}(m, n, d)$ denote the distribution of random $d$-CNFs on $n$ variables obtained by sampling $m$ clauses (out of the ${n \choose d}2^d$ possible clauses) uniformly at random with replacement.
\end{defn}
The proof is delayed to Section \ref{sec:truly-random}; to get a feeling for the proof, we first prove an easier lower bound for a simpler distribution of \emph{balanced} random CNFs.

\subsection{Balanced Random CNFs}

\begin{defn}
  Let $X = \{x_1, \ldots, x_n\}$ and $Y = \{y_1, \ldots, y_n\}$ be two disjoint sets of variables, and the distribution ${\cal F}(m, n, d)^{\otimes 2}$ denotes the following distribution over $2d$-CNFs:
  First sample \[\mathcal{F}^1 = C_1^1 \wedge C_2^1 \wedge \cdots \wedge C_m^1\] from ${\cal F}(m, n, d)$ on the $X$ variables, and then \[\mathcal{F}^2 = C_1^2 \wedge C_2^2 \wedge \cdots \wedge C_m^2\] from ${\cal F}(m, n, d)$ on the $Y$ variables independently.
  Then output \[\mathcal{F} = (C_1^1 \vee C_1^2) \wedge (C_2^1 \vee C_2^2) \wedge \cdots \wedge (C_m^1 \vee C_m^2).\]
\end{defn}

This distribution shares the well-known property with ${\cal F}(m, n, d)$ that dense enough formulas are unsatisfiable with high probability.

\begin{lem}
  Let $c > 2/\log e$ and let $n$ be any positive integer.
  If $d \in [n]$ and $m \ge c n 2^{2d}$ then $\mathcal{F} \sim {\cal F}(m, n, d)^{\otimes 2}$ is unsatisfiable with high probability.
\end{lem}
\begin{proof}
  Fix any assignment $(x, y)$ to the variables of ${\cal F}$.
  The probability that the $i$th clause is satisfied by the joint assignment is $1-1/2^{2d}$, and so the probability that \emph{all} clauses are satisfied by the joint assignment is $(1-1/2^{2d})^m \le e^{-m/2^{2d}}$, since the clauses are sampled independently.
  By the union bound, the probability that some joint assignment satisfies the formula is at most $2^{2n}e^{-m/2^{2d}} = 2^{2n - (\log e) m/ 2^{2d}} \le 2^{2n - (\log e) c n} \le 2^{-\Omega(n)}$.
  Thus, the probability that the formula is unsatisfiable is at least $1-2^{-\Omega(n)}$.
\end{proof}

The main theorem of this section is that ${\cal F} \sim {\cal F}(m, n, d)^{\otimes 2}$ require large $\CC$- and $\RCC$-proofs, which is obtained by using Theorem \ref{thm:main_theorem} and applying the well-known method of symmetric approximations \cite{berg:symmetric, haken-cook} to obtain lower bounds on monotone circuits computing the function $\CSP_{\Search({\cal F})}$.
We use the following formalization of the method which is exposited in Jukna's excellent book \cite{jukna-book}.
First we introduce some notation: if $U \subseteq \set{0,1}^N$, then for $r \in [N]$ and $b \in \set{0,1}$ let \[ A_b(r, U) = \max_{I \subseteq [n]: |I| = r} |\set{u \in U \st \forall i \in I: u_i = b}|.\]

% \begin{thm}[Theorem 9.18 in Jukna]\label{thm:sym_method_of_approx}
%   Let $f: \set{0,1}^N \rightarrow \set{0,1}$ be a monotone boolean function and let $1 \leq r, s \leq N$ be any positive integers.
%   Let $U \subseteq f^{-1}(1)$ and $V \subseteq f^{-1}(0)$ be arbitrary subsets of accepting and rejecting inputs of $f$.
%   Then every monotone circuit that outputs $1$ on all inputs in $U$ and $0$ on all inputs in $V$ has size at least  \[ \min \left\{\frac{|U| - (s-1)A_1(1, U)}{(s-1)^r A_1(r, U)}, \frac{|V|}{(r-1)^sA_0(s, V)} \right\}.\]
% \end{thm}
% With a very small loss in the parameters, Jukna also proved the following version of the above theorem which holds for monotone real circuits.

\begin{thm}[Theorem 9.21 in Jukna]\label{thm:real_sym_method_of_approx}
  Let $f: \set{0,1}^N \rightarrow \set{0,1}$ be a monotone boolean function and let $1 \leq r, s \leq N$ be any positive integers.
  Let $U \subseteq f^{-1}(1)$ and $V \subseteq f^{-1}(0)$ be arbitrary subsets of accepting and rejecting inputs of $f$.
  Then every real monotone circuit that outputs $1$ on all inputs in $U$ and $0$ on all inputs in $V$ has size at least \[ \min \left\{\frac{|U| - (2s)A_1(1, U)}{(2s)^{r+1} A_1(r, U)}, \frac{|V|}{(2r)^{s+1}A_0(s, V)} \right\}.\]
\end{thm}

Next we state the main theorem of this section.

\begin{thm}\label{thm:random_cnf_monotone_lb}
  Let $d = 4 \log n$ and $m = cn^2 2^d$ where $c > 2/\log e$ is some constant.
  Let ${\cal F} \sim {\cal F}(m, n, d)^{\otimes 2}$ with variable partition $(X, Y)$, and let \[U = {\cal U}(\set{0,1}^X), V = {\cal V}(\set{0,1}^Y).\]
  Then with high probability any real monotone circuit separating $U$ and $V$ has at least $2^{\tilde \Omega(n)}$ gates.
%  In particular, setting $d = 4 \log n$ and $m = n^2 2^d$  for every $\epsilon \in (0,1/2)$ the real monotone circuit size is at least $2^{\tilde \Omega(n)}$.
\end{thm}
\begin{cor}
  Let $n$ be a sufficiently large positive integer, and let $d = 4 \log n, m = n^6$.
  If $\mathcal{F} \sim {\cal F}(m, n, d)^{\otimes 2}$ then with high probability every $\RCC_1$-refutation (and therefore, Cutting Planes refutation) of ${\cal F}$ has at least $2^{\tilde \Omega(n)}$ lines.
\end{cor}
\begin{proof}
  Immediate consequence of Theorems \ref{thm:main_theorem} and \ref{thm:random_cnf_monotone_lb}.
\end{proof}

% \begin{defn}
% Let $G = (L \cup R, E)$ be a given bipartite graph, and $S \subseteq L$. We say that $S$ expands by a factor $r$ if the neighborhood of $S$, $N(S)$, satisfies $|N(S)| \ge r |S|$.
% \end{defn}

% \begin{defn}
% A bipartite graph $G=(L \cup R, E)$ is $d$-left regular if every vertex on the left side has degree exactly $k$.
% \end{defn}

The proof of Theorem \ref{thm:random_cnf_monotone_lb} is rather straightforward, and comes down to the essential property that random $d$-CNFs are good expanders.
The next lemma records the expansion properties we require of random CNFs; the proof is adapted from the notes of Salil Vadhan \cite{vadhan-book}.

\begin{lem}\label{lem:expansion}
  Let $0 < \varepsilon < 1$ be arbitrary, and let $n$ be any sufficiently large positive integer.
  Let $d = 4\log n$, $m = n^22^d$, and sample ${\cal F} \sim {\cal F}(m, n, d)$.
  For any subset $S \subseteq {\cal F}$ of clauses let $\vars(S)$ denote the subset of variables appearing in any clause of $S$.
  Any set $S \subseteq {\cal F}$ of size $s \leq n/ed^{2}$ satisfies \[|\vars(S)| \geq (1- \varepsilon)ds\] with high probability.
%   Let $G=(L \cup R, E)$ be a $d$-left regular random bipartite graph generated from the following distribution.
%   For each vertex on the left let its neighborhood be a random subset of $R$ of size $d$ chosen uniformly and independently of other choices.
%   Let $m = |L|$ and $n = |R|$.
%   Given $r$ and $\ell$, the probability that every set of $L$ of size $\ell$ expands by a factor $r$ is at least
% \[ 1 - \left( \frac{m e}{\ell} \right)^{\ell} \left( \frac{ \ell d^2 e} {(n-d)(d-r)}\right)^{\ell(k-r)}.\]
\end{lem}
\begin{proof}
  Fix any set $S \subseteq {\cal F}$ of size $s$, and for each clause $C \in S$ sample the variables in $C$ one at a time without replacement.
  Let $v_1, v_2, \ldots, v_{ds}$ denote the concatenation of all sequences of sampled variables over all $C \in S$.
  We say that variable $v_i$ is a repeat if it has already occurred among $v_1, \ldots, v_{i-1}$.
  In order for $|\vars(S)| < (1-\varepsilon) ds$ the concatenated sequence must have at least $\varepsilon ds$ repeats, and the probability that variable $v_i$ is a repeat is at most $(i-1)/n \le ds/n$.
  This implies that \[\Pr[|\vars(S)| < (1-\varepsilon)ds] \leq {ds \choose \varepsilon ds}\left( \frac{ds}{n} \right)^{\varepsilon ds} \leq \left( \frac{eds}{\varepsilon ds} \right)^{\varepsilon ds} \left( \frac{ds}{n} \right)^{\varepsilon ds} \leq \left( \frac{1}{\varepsilon d} \right)^{\varepsilon ds}\] using standard bounds on binomial coefficients and the fact that $s \leq n/ed^2$.
  Thus \[\Pr[\exists S: |S| = s, |\vars(S)| < (1-\varepsilon)ds] \leq m^s \left( \frac{1} {\varepsilon d}\right)^{\varepsilon ds},\] and since $m = n^2 2^d$ and $d = 4 \log n$ we get that \[ s \log m \ll \varepsilon ds \log \varepsilon d\] for sufficiently large $n$, implying the previous probability is $o(1)$.
\end{proof}

Using the expansion lemma we are ready to prove Theorem~\ref{thm:random_cnf_monotone_lb}.

\begin{proof}[Proof of Theorem~\ref{thm:random_cnf_monotone_lb}]
  We shall apply Theorem~\ref{thm:real_sym_method_of_approx} to $U = {\cal U}(\set{0,1}^n)$ and $V = {\cal V}(\set{0,1}^n)$ (cf.~Section \ref{sec:csp-sat}) with $r = s = n/ed^2$.
  Recall that ${\cal U}$ and ${\cal V}$ are the functions mapping $x$ inputs to $1$-inputs of $\CSP_{\Search({\cal F})}$ and mapping $y$ inputs to $0$-inputs of $\CSP_{\Search(\mathcal{F})}$, respectively.
  To finish the argument we need to compute $|U|, A_1(1,U), A_1(r,U), |V|, A_0(s,V)$.

  It is easy to see that every variable participates in some clause in $\mathcal{F}$ with high probability.
  This implies that ${\cal U}$ is one-to-one with high probability, and thus $|U| = 2^n$ with high probability.

  Recall that the $0$-inputs of $\CSP_{\Search(\mathcal{F})}$ correspond to substituting $y$-assignment into $\mathcal{F}$ and writing out truth tables of the all the clauses.
  The truth tables corresponding to the clauses that were satisfied by the $y$-assignment are identically 1, and the truth tables corresponding to the clauses that were not satisfied by the given $y$-assignment contain exactly one $0$-entry.
  Given a $y$ assignment we call the set of clauses that were not satisfied the assignment the \emph{profile} of $y$.
  The next lemma implies that the profiles of all $y$-assignments are distinct with high probability.
  \begin{lem}\label{lem:distinct_profiles}
    Let $\mathcal{F} \sim {\cal F}(m, n, d)$, and define the following $2^{n} \times m$ matrix $M$, with the rows labelled by assignments $ \alpha \in \set{0,1}^n$ and the columns are labelled by clauses of ${\cal F}$.
    Namely, for any pair $(\alpha, i)$ set \[M[\alpha, i] =
      \begin{cases}
        1 & \text{ if the } i \text{th clause is not satisfied by } \alpha, \\
        0 & \text{ otherwise}.
      \end{cases}
     \]
    For any $c > 2/ \log e$, if $m \ge c 2^d n^2/d$ then the rows of $M$ are distinct with high probability.
  \end{lem}
  \begin{proof}
    We think of $M$ as generated column by column with the columns sampled independently.
    Fix two assignments $\alpha$ and $\widehat{\alpha}$ such that $\alpha \neq \widehat{\alpha}$.
    Let $S$ be the set of indices on which the two assignments differ, i.e., $S = \{i \mid \alpha_i \neq \widehat{\alpha}_i\}$.
    Set $s = |S|$.
    Let $C_i$ denote the $i$th clause, and we say that $C_i$ \emph{overlaps} $S$ if $C_i$ contains a variable in $S$.
    %Let $E$ be the event ``$C_i$ is unsat by $\widehat{\alpha}$''; $H$ be the event ``$C_i$ is sat by $\alpha$''; $D$ be the event ``$C_i$ overlaps $S$''.
    Then
    \begin{align*}
      \Pr[C_i \text{ unsat by } \widehat{\alpha} \text{ and satisfied by } \alpha] & = \frac{1}{2^d} \left( 1 - \frac{{n-s \choose d}}{{n \choose d}} \right)\\
                  &\ge \frac{1}{2^d} \frac{{n \choose d} - {n-1 \choose d}}{{n \choose d}} = \frac{1}{2^d} \frac{{n-1 \choose d-1}}{{n \choose d}} = \frac{d}{2^d n}.\\
    \end{align*}
    Thus the probability that rows $\alpha$ and $\widehat{\alpha}$ agree on column $i$ is at most $1 - \frac{d}{2^d n}$.
    Since columns are sampled independently, the probability that $\alpha$ and $\widehat{\alpha}$ agree on all columns is at most $\left( 1 - \frac{d}{2^d n} \right)^m \le e^{-dm/(2^d n)}$.
    By a union bound over ordered pairs of assignments, the probability that there exists a pair of rows that agree on all columns is at most $2^{2n}e^{-dm/(2^d n)} = 2^{2n - (\log e) dm / (2^d n)} \le 2^{2n - (\log e) c n} = 2^{-\Omega(n)}$.
    Thus, the probability that all columns are distinct is at least $1-2^{-\Omega(n)}$.
  \end{proof}

  Since each profile is distinct with high probability, this implies that ${\cal V}$ is 1-1 with high probability, and therefore $|V| = 2^n$.
  It remains to bound the terms $A_1(1, U), A_1(r, U),$ and $A_0(s, V)$.

  \spcnoindent
  {\bf Bounding $A_1(1,U)$.} Fixing a single bit of a 1-input in $U$ to $\CSP_{\Search(\mathcal{F})}$ to 1 is the same as selecting a vertex $C$ in the bipartite constraint graph of $\Search(\mathcal{F})$ and an assignment $\alpha$ to the variables which participate in $C$, and then setting $\TT_C(\alpha) = 1$.
  By the definition of ${\cal U}$, any input $x \in \set{0,1}^n$ fixing this bit to $1$ determines $d$ out of $n$ variables of $x$ exactly.
  Thus the number of $x \in \set{0,1}^n$ that are consistent with this partial assignment is $2^{n-d}$, and since ${\cal U}$ is one-to-one, we have $A_1(1,U) = 2^{n-d}$.

  \spcnoindent
  {\bf Bounding $A_1(r,U)$.} Similar to the previous bound, but now we fix $r$ of the truth table bits to $1$.
  By definition of ${\cal U}$, these bits must be chosen from $r$ distinct truth tables in the $1$-input in order to be consistent with any $x \in \set{0,1}^n$.
  With respect to the underlying CNF ${\cal F}$, this corresponds to fixing an assignment to the set of variables appearing in an arbitrary set ${\cal S}$ of $r$ clauses in ${\cal F}$.
  By Lemma~\ref{lem:expansion}, with high probability we have $|\vars(S)| \geq ds/2$.
  Thus fixing these $r$ bits in the definition of $A_1(r, U)$ corresponds to setting at least $rd/2$ of the input variables participate in the constraints with determined truth tables.
  The number of $x$ inputs that are consistent with these indices fixed is therefore $\le 2^{n-rd/2}$, and so $A_1(r,U) \le 2^{n-rd/2}$.

  \spcnoindent
  {\bf Bounding $A_0(s,V)$.} This case is similar to $A_1(r,U)$. We get $A_0(s,V) \le 2^{n-sd/2}$.

  Observe that $(s-1)A_1(1,U) = (s-1)2^{n-d} = (s-1)2^n/n^2 \le 2^{n-1}$.
  Putting this altogether we get the following lower bound on monotone circuit size is at least \[\frac{2^{n-1}}{(s-1)^s 2^{n-sd/2}} = 2^{sd/2 - s \log (s-1)} \ge 2^{s(d/2- \log s)} \ge 2^{\tilde \Omega(n)},\] where the last inequality follows from $s = n/ed^{2}$ and $d/4 \ge \log n$.
\end{proof}

%%% Local Variables:
%%% mode: latex
%%% TeX-master: "paper"
%%% End:

\subsection{Uniformly Random CNFs} \label{sec:truly-random}

In this section we show how to modify the argument from the previous section to apply to the ``usual'' distribution of random CNFs ${\cal F}(m, n, d)$.
Our approach is simple: using the probabilistic method we find a partition of the variables of a random formula ${\cal F} \sim {\cal F}(m,n ,d)$ such that many of the clauses in ${\cal F}$ are balanced with respect to the partition.
Ideally, every clause would be so balanced, but it turns out that this is too strong --- instead, we show that we can balance many of the clauses, and the imbalanced clauses that remain are always satisfied by a large collection of assignments.
First we introduce our notion of ``imbalanced'' clauses.
\begin{defn}
  Fix $\epsilon > 0$. Given a partition of $n$ variables into $x$-variables and $y$-variables, clause $C$ is called $X$-heavy if it contains more than $(1-\epsilon)d$ $x$-variables. Clause $C$ is called $Y$-heavy if it contains more than $(1-\epsilon)d$ $y$-variables. Clause $C$ is called balanced if it is neither $X$-heavy nor $Y$-heavy.
\end{defn}

We recall some basic facts from probability theory which will be used in our main lemma.
\begin{lem}[Lov\'{a}sz Local Lemma]\label{lemma:lll}
  Let $\mathcal{E} = \{E_1, \ldots, E_n\}$ be a finite set of events in the probability space $\Omega$. For $E \in \mathcal{E}$ let $\Gamma(E)$ denote the set of events $E_i$ on which $E$ depends. If there is $q \in [0,1)$ such that $\forall E \in \mathcal{E}$ we have $\Pr(E) \le q (1-q)^{|\Gamma(E)|}$, then the probability of avoiding all sets $E_i$ is at least $\Pr(\overline{E_1} \wedge \overline{E_2} \wedge \cdots \wedge \overline{E_n}) \ge (1-q)^n$.
\end{lem}

\begin{fact}[Entropy bound on binomial tail]\label{fact:entropy_bound}
  Given $\epsilon > 0$ we have
  \[ \sum_{j = 0}^{\lfloor \epsilon n \rfloor} {n \choose j} \le e^{n H(\epsilon)},\]
  where $H(\epsilon) = - \epsilon \log \epsilon - (1-\epsilon) \log (1-\epsilon)$ is the binary entropy function.
\end{fact}

\begin{fact}[Multiplicative Chernoff Bound]\label{fact:chernoff}
  Suppose $Z_1, \ldots, Z_n$ are independent random variables taking values in $\{0,1\}$. Let $Z$ denote their sum and let $\mu = \mathbb{E}(Z)$ denote the sum's expected value. Then for any $\delta \in (0,1)$ we have
  \[ \Pr(X \ge (1+\delta) \mu) \le e^{-\delta^2 \mu/3}.\]
\end{fact}

We now prove the main lemma of this section, which shows that for ${\cal F} \sim {\cal F}(m, n, d)$ a good partition of the variables exists with high probability.

\begin{lem}
  Let $\mathcal{F} \sim \mathcal{F}(m,n,d)$ where $d = c \log n$ and $m = \poly(n)$.
  There exists a partition of the variables of ${\cal F}$ into two sets $(X, Y)$ such that the following holds:
  \begin{enumerate}
  \item The number of $X$-heavy clauses and $Y$-heavy clauses are each upper bounded by $m' = m 2^{-(1-(\log e) H(\epsilon))d+1}$.
  \item There exists a set $U'$ of $2^{n/2 \pm o(n)}$ truth assignments to the $X$ variables satisfy all $X$-heavy clauses, and similarly a set $V'$ of $2^{n/2 \pm o(n)}$ truth assignments to the $Y$-variables satisfying all of the $Y$-heavy clauses.
  \end{enumerate}
\end{lem}
\begin{proof}
  We prove the existence of such a partition by the probabilistic method.
  For each variable, flip a fair coin and place it in $X$ if the coin is heads and in $Y$ otherwise.
  Let $Z_i$ be the random variable indicating whether clause $i$ is $X$-heavy.
  Then \[ \Pr(Z_i = 1) = \sum_{j = 0}^{\epsilon d} { d \choose j} 2^{-d} \le 2^{-d} e^{-d H(\epsilon)} \le 2^{-(1-(\log e) H(\epsilon))d},\]
  where the inequality follows from Fact~\ref{fact:entropy_bound}.
  Let $Z = \sum_{i=1}^m Z_i$; then $\mathbb{E}(Z) \le m 2^{-(1-(\log e) H(\epsilon))d} = m'$.
  By the multiplicative Chernoff bound (see Fact~\ref{fact:chernoff}) we have \[ \Pr(Z > (3/2) m 2^{-(1-(\log e) H(\epsilon))d}) \le e^{-\frac{m 2^{-(1-(\log e) H(\epsilon))d}}{12}},\]
  % By adjusting $\epsilon'$ as necessary (note $k = \omega(1)$) 
  and we thus have $Z \le m'$ with high probability.
  An identical calculation applies for the $Y$-heavy clauses.

  Next, let $W_i$ be the random variable indicating whether a given fixed variable occurs in clause $i$ and clause $i$ is $X$-heavy and let $W = \sum_i W_i$.
  Then $\Pr(W_i = 1) \le 2^{-(1-(\log e) H(\epsilon))d} d/n$.
  By the multiplicative Chernoff bound (see Fact~\ref{fact:chernoff}) we have
  \[ \Pr(W > (3/2) m 2^{-(1-(\log e) H(\epsilon))d}d/n) \le e^{-\frac{m 2^{-(1-(\log e) H(\epsilon))d}d/n}{12}}.\]
  We conclude that $W \le m'd/n$ whp, and an identical calculation again holds for the $Y$-heavy clauses.

  Noting that the number of $x$-variables is $n/2 \pm o(n)$ with high probability, by the probabilistic method we choose a partition $(X, Y)$ which satisfies each of the above properties (the bound on $Z$ and $W$, and achieving near balance in the $X$ and $Y$ variables), and note that such a partition exists with high probability over ${\cal F}(m, n, d)$.
  With this partition fixed, consider selecting a random assignment to the $X$-variables.
  Let $E_i$ be the event that $X$-heavy clause $i$ is falsified by the random assignment, and observe that $\Pr(E_i) \le 2^{-(1-\epsilon)d}$ since the clause is $X$-heavy.
%  Let $d' = m2^{-(1-(\log e) H(\epsilon))d+1}d/n$ and $d = d' n/d$.
  Then the number of events $E_i$ is at most $m'$, and for any event $E_i$ the number of events that share an $x$-variable with $E_i$ is at most $m'd^2/n$.
  Set $q = n/(100 m'd)$.
  Then for each $E_i$ we have  \[q(1-q)^{|\Gamma(E_i)|} \ge q e^{-2q m'd^2/n } \ge \frac{n}{100 dm'} e^{-d/50} \ge 2^{-(1-\epsilon)d},\]
  provided $d \ge c \log n$ for a big enough constant $c$.
  Applying Lov\'{a}sz Local Lemma (see Lemma~\ref{lemma:lll}) we get that probability that an assignment satisfies all $X$-heavy clauses is at least
  \[ (1-q)^{m'} \ge (1-n/(100 dm'))^{m'} \ge e^{-n/(50d)}.\]
  Thus the number of assignments to the $X$-variables satisfying all heavy clauses is at least $2^{n/2 \pm o(n)}$, and an identical calculation applies to the $Y$ variables.
\end{proof}
Now we will do the whole argument with respect to ${\cal U}(U')$ and
${\cal V}(V')$ chosen from the previous lemma.
The reason that this works is that since every $\alpha \in U'$ satisfies all $X$-heavy clauses,
if we look at a subset $S$ of the variables of the monotone CSP
that are set to false and count the number of maxterms that are
consistent with it, the count is nonzero only when {\it none} of these variables come
from a $X$-heavy clause. Similarly if we look at a subset
$S$ of the variables of the monotone CSP that are set to
true and count the number of maxterms that are consistent with it,
this count is nonzero only when {\it none} of these variables come
from an $X$-heavy clause. Therefore, when we calculate
$A_1(s, {\cal U}(U'))$ and $A_0(s,{\cal V}(V'))$, the calculation
is with respect to the $X$-balanced clauses, and $Y$-balanced
clauses, respectively. Thus we can use the same expansion calculation
that we already did.

There is also a minor modification required in order to argue that
${\cal V}$ is one-to-one when restricted to $V'$. As before, it suffices to show that
for any two assignments $\alpha,\beta$ in $V'$, that the probability
that they agree on all of the balanced clauses is very small,
and then take a union bound over all of the {\it balanced} clauses.
This calculation is nearly identical to the one that we already did,
but now the union bound is over the number of balanced clauses,
which is at least half of all clauses, so the calculation is
essentially the same.

With the above modifications, the arguments from the previous section imply the next theorem.
\begin{thm}
  Let $n$ be a sufficiently large positive integer.
  Let ${\cal F} \sim {\cal F}(m, n, d)$ for $m = \poly(n)$ and $d = c \log n$ for a large universal constant $c$.
  With high probability, there exists a partition $(X, Y)$ of the variables of ${\cal F}$ and an $\varepsilon > 0$ such that the search problem $\Search({\cal F})$ defined with respect to this partition satisfies the following: any real monotone circuit computing $\CSP_{\Search({\cal F})}$ requires at least $2^{\Omega(n^\varepsilon)}$ gates.
\end{thm}
\begin{cor}
  Let ${\cal F}$ be distributed as above.
  There exists $\varepsilon > 0$ such that with high probability any $\RCC_1$-refutation requires $2^{\Omega(n^\varepsilon)}$ lines.
\end{cor}

%%% Local Variables:
%%% mode: latex
%%% TeX-master: "paper"
%%% End:

%\section{Real Protocols}
%\input{real_protocols.tex}

% \section{Lower Bounds for Lifted Tseitin}
% \input{lifted_tseitin.tex}

\bibliography{biblio}
\bibliographystyle{plain}

\end{document}